\documentclass[runningheads]{llncs}
\usepackage{latexsym}
\input epsf

\begin{document}

\title{Faster Deterministic Algorithms for Packing, Matching and $t$-Dominating Set Problems}

\titlerunning{Packing, Matching and $t$-Dominating Set Problems}

\author{Shenshi Chen \and Zhixiang Chen}

\authorrunning{S. Chen \and Z. Chen}

\institute{Department of Computer Science\\
University of Texas-Pan American\\ Edinburg, TX 78539, USA\\
\email{schen@broncs.utpa.edu}, \email{zchen@utpa.edu}}

\maketitle

\begin{abstract}
In this paper, we devise three deterministic algorithms for
solving the $m$-set $k$-packing,
$m$-dimensional $k$-matching, and  $t$-dominating set  problems
in time $O^*(5.44^{mk})$, $O^*(5.44^{(m-1)k})$ and $O^*(5.44^{t})$, respectively.
Although recently there have been remarkable progresses on randomized solutions to those problems,
yet our bounds make good improvements on the best known bounds for deterministic solutions to those
problems.
\\
\\
{\bf Keywords:} Packing; matching; dominating sets; group algebra; monomial testing;
randomized algorithms; derandomization.
\end{abstract}

\section{Introduction}

\subsection{The $m$-Set $k$-Packing Problem}
Let $\mathcal{ S}$ be a collection of sets so that each member in $\mathcal{ S}$ is a subset of an $n$-element set
$U$. In addition, members in $\mathcal{ S}$ have the same size $m\ge 3$. The $m$-set $k$-packing problem asks
whether there are $k$ members in $\mathcal{ S}$ such that those members are pairwise disjoint.
It is known that this problem is $W[1]$-complete with respect to the parameter $k$ \cite{fellows99}.
\vspace{-4mm}
\begin{table}\label{tab-1}
\center
\caption{Algorithms for $m$-set $k$-packing}
\vspace{-2mm}
\begin{tabular}{|l|l|l|}\hline
~~References & ~~Randomized & ~~Deterministic \\
           & ~~Algorithms           & ~~Algorithms \\ \hline
~~Jia {\em et al.} \cite{jia04}~~&        &~~$\displaystyle O^*(m^k(g(m,k))^{mk})$,~~ \\
                             &        &~~ where $g(m,k)$ is linear in $mk$.~~ \\ \hline
~~Koutis \cite{koutis05}~~ & ~~$\displaystyle O^*(10.88^{mk})$~~ & ~~$\displaystyle O^*(25.6^{mk})$ ~(see \cite{jianer-chen07,liu06})~~ \\ \hline
~~Fellows {\em et al.} \cite{fellow08}~~  &       & ~~$\displaystyle O^*(13.78^{mk})$ ~(see \cite{jianer-chen07,liu06})~~ \\ \hline
~~Koutis \cite{koutis08}~~ &~~ $\displaystyle O^*(2^{mk})~~$ & \\ \hline
~~Bj\"{o}rklund {\em et al.} \cite{Bjorklund2010b}~~  &~~ $O^*(f(m,k))$ ~~&  \\ \hline
~~Chen \cite{bill13}~~ & & ~~$\displaystyle O^*(12.8^{mk})$~~\\ \hline
~~This paper (Theorem \ref{packing-thm})~~ &  &~~ $\displaystyle O^*(5.44^{mk})$ ~~\\ \hline
\end{tabular}
\end{table}
\vspace{-4mm}
 In literature,
most work has been done for the special case of $m=3$ (see, \cite{jianer-chen07,liu06,Bjorklund2010b} for an overview of
the work). For the general $m$-set $k$-packing problem, Table \ref{tab-1} summarizes the previous algorithms and
compares them with our result that improves the best known deterministic time bounds.

In Chen {\em et al.} \cite{jianer-chen07} and Liu {\em et al.} \cite{liu06}
there are detailed discussions about many previous algorithms for the $m$-set $k$-packing problem,
especially for the case of $m=3$. For example, the deterministic algorithm by Koutis \cite{koutis05} has time
$O^*(2^{O(mk)})$. It was pointed out in  \cite{jianer-chen07,liu06} that this bound has large constants.
E.g., when $m=3$, the bound is at least $O^*(32000^{3k})$.
It was also pointed out in the two papers that this bound can be improved to $O^*(25.6^{3k})$ by a new perfect hashing technique.
Furthermore, this new hashing technique can improve the general bound by Koutis to  $O^*(25.6^{mk})$, which is listed in Table \ref{tab-1}.
For the deterministic algorithm by Fellows {\em et al.} \cite{fellow08}, the original bound is $exp(O(mk))$ with large constants.
When $m=3$, it was pointed out in \cite{jianer-chen07,liu06} that this bound is $O^*((12.7D)^{3k})$ for $D\ge 10.4$ and that
this bound can be further improved to $O^*(13.78^{mk})$, which is listed  in Table \ref{tab-1},
by the aforementioned perfect hashing technique.

As noted by Bj\"{o}klund {\em et al.} \cite{Bjorklund2010b}, the $O^*(f(m,k))$ bound of their randomized algorithm
is most efficient for small $m$. E.g., when $m=3$,  $O^*(f(3,k)) = O^*(1.4953^{3k})$. In general, this bound does not give
a $O^*((2-\epsilon)^{mk})$ bound for some small value $0<\epsilon <1$.

\subsection{The $m$-dimensional $k$-Matching Problem}

Let $U_1, U_2,\ldots U_m$ be pairwise disjoint sets and $U = U_1\times U_2\times \cdots \times U_m$.
The $m$-dimensional $k$-matching problem asks, for any  given set $\mathcal{C} \subseteq U$ of $m$-component tuples,
whether $\mathcal{C}$ contains a size $k$ subset $\mathcal{C'}$ such that
all tuples in $\mathcal{C'}$ are mutually disjoint, i.e., any two tuples in $\mathcal{C'}$ have no common components.
It is well known that the $3$-dimensional $k$-matching problem is a classical NP-complete problem.
Like the $m$-set $k$-packing problem, most work in literature has concentrated on the special case of $m=3$.
E.g., Liu {\em et al.} \cite{liu06} obtained the best deterministic time bound $O^*(2.77^{3k})$ for $3$-dimensional $k$-packing.
\cite{jianer-chen07,liu06,Bjorklund2010b} have good overviews of the research about this problem.
For general $m$-dimensional and $k$-packing Table \ref{tab-2} summarizes the previous algorithms
in comparison with our result that improves the best known deterministic time bounds.
\vspace{-4mm}
\begin{table}\label{tab-2}
\center
\caption{Algorithms for $m$-dimensional $k$-matching}
\vspace{-2mm}
\begin{tabular}{|l|l|l|}\hline
~~References & ~~Randomized & ~~Deterministic \\
           & ~~Algorithms           & ~~Algorithms \\ \hline
~~Downey and Fellows \cite{fellow08}~~&        &~~$\displaystyle O^*((mk)!(mk)^{3mk+1})$~~ \\ \hline
~~Koutis \cite{koutis05}~~ & ~~$\displaystyle O^*(10.88^{mk})$~~ & ~~$\displaystyle O^*(25.6^{mk})$ ~(see \cite{jianer-chen07,liu06})~~ \\ \hline
~~Fellows {\em et al.} \cite{fellow08}~~  &       & ~~$\displaystyle O^*(13.78^{mk})$ ~(see \cite{jianer-chen07,liu06})~~ \\ \hline
~~Koutis \cite{koutis08}~~ &~~ $\displaystyle O^*(2^{mk})~~$ & \\ \hline
~~Koutis and Williams \cite{koutis09}~~&~~$\displaystyle O^*(2^{(m-1)k})$~~ & \\ \hline
~~Bj\"{o}rklund {\em et al.} \cite{Bjorklund2010b}~~  &~~ $\displaystyle O^*(2^{(m-2)k})$ ~~&  \\ \hline
~~This paper (Theorem \ref{matching-thm})~~ &  &~~ $\displaystyle O^*(5.44^{(m-1)k})$ ~~\\ \hline
\end{tabular}
\end{table}
\vspace{-4mm}

As mentioned in the previous subsection, Chen {\em et al.} \cite{jianer-chen07} and Liu {\em et al.} \cite{liu06}
presented detailed discussions about many algorithms for the $m$-dimensional $k$-matching problem,
especially for the case of $m=3$.  The deterministic time bounds listed in the table
for Koutis \cite{koutis05} and Fellows {\em et al.} \cite{fellow08} are not the original bounds in the respective papers.
Instead, those are pointed out in \cite{jianer-chen07} and \cite{liu06} (see the above subsection).

\subsection{The $t$-Dominating Set Problem}

This problem is a partial satisfaction variant of the dominating set problem, an classical NP-complete problem.
Given any simple undirected graph $G=(V,E)$, for any nodes $u, v \in V$, we say that $u$ dominates $v$ if either $u=v$ or
$(u,v) \in E$. The $t$-dominating set problem asks us to find a minimal number $k$ such that
 there exists a size $k$ set $S\subseteq V$ that dominates at least $t$ distinct nodes. When $t= |V|$, this problem
becomes the dominating set problem, hence it is $W[2]$-complete \cite{fellow08} with respect to the parameter $k$.
Table \ref{tab-3} summarizes the previous algorithms for this problem
and compares them with our result that improves the best known deterministic time bound.
\vspace{-4mm}
\begin{table}\label{tab-3}
\center
\caption{Algorithms for $t$-dominating sets}
\vspace{-2mm}
\begin{tabular}{|l|l|l|}\hline
~~References & ~~Randomized & ~~Deterministic \\
           & ~~Algorithms           & ~~Algorithms \\ \hline
~~Kneis {\em et al.}\cite{kneis07}~~&~~$\displaystyle O^*((4+\epsilon)^t)$     &~~$O^*((16+\epsilon)^t)$,~~ \\ \hline
~~Koutis and Williams \cite{koutis09}~~&~~$\displaystyle O^*(2^t)$~~ & \\ \hline
~~This paper (Theorem \ref{dominating-thm})~~ &  &~~ $\displaystyle O^*(5.44^{t})$ ~~\\ \hline
\end{tabular}
\end{table}
\vspace{-4mm}

\subsection{Techniques}

Our techniques are built upon derandomizing a randomized group algebraic approach to testing $q$-monomials in
a multivariate polynomial represented by a formula. Group algebraic approach to testing multilinear
monomials is initiated by Koutis \cite{koutis08} and further developed by Williams \cite{williams09}.
Later progresses on testing multilinear monomials and $q$-monomials in multivariate polynomials can be found in \cite{koutis09,chen11,chen11a,chen11b,chen12a,chen12b,bill12,bill13}.
Randomized algebraic techniques have recently led to the once fastest
randomized algorithms of time $O^*(2^k)$ for the  $k$-path
problem and other problems \cite{koutis08,williams09}. Another recent
remarkable example is the improved $O(1.657^n)$ time randomized algorithm for the Hamiltonian
path problem by Bj\"{o}rklund~\cite{Bjorklund2010}. 
 Bj\"{o}rklund {\em et al.} further extended the above randomized algorithm
to the $k$-path testing problem with $O^*(1.657^k)$ time complexity \cite{Bjorklund2010b}.

Chen in \cite{bill13} designed an $O^*(2^k)$ time randomized algorithm for solving
the $q$-monomial testing problem for polynomials represented by circuits,
regardless of the primality of $q\ge 2$ and derived an  $O^*(12.8^k)$ deterministic algorithm
for polynomials represented by formulas. The second algorithm is devised by derandomizing
the first algorithm with the help of the perfect hashing functions by Chen {\em et al.}
\cite{jianer-chen07} and the deterministic polynomial identity testing algorithm by Raz
and Shpilka \cite{raz05} for noncommunicative polynomials.

Our approach is to further improve the $O^*(12.8^k)$ time bound of the deterministic algorithm in \cite{bill13}.
Besides algebraic replacements, we give a formal proof about how to use Walsh-Hadamard transformation
to speed up the multiplication of two group algebraic elements. We present a simple and efficient
deterministic algorithm to solve the polynomial identity testing problem for a special case of read-once formulas.
We also use a near optimal family of  perfect hashing functions  by Naor {\em et al.} \cite{naor95} to assist the
derandomization process. We obtain the following two results:

\begin{theorem}\label{thm-1}
Let $q\ge 2$ be a fixed integer. Let $F(x_1,x_2,\ldots,x_n)$ be an
$n$-variate polynomial represented by a formula $\mathcal{ C}$
of size $s(n)$. There is a deterministic $O^*(5.44^ks^4(n))$ time
algorithm to test whether $F$ has a $q$-monomial of degree $k$ in
its sum-product expansion.
\end{theorem}

\begin{corollary}\label{cor}
Let $q\ge 2$ be a fixed integer. Let $F(x_1,x_2,\ldots,x_n, z)$ be a
$(n+1)$-variate polynomial represented by a formula $\mathcal{ C}$
of size $s(n)$. There is a deterministic $O^*(t^2 5.44^ks^4(n))$ time
algorithm to test whether $F$ has a monomial $z^t \pi$ in
its sum-product expansion such that $\pi$ is a $q$-monomial of $x$-variables with degree $k$.
\end{corollary}

The above results will be proved in Section \ref{FDTM-section}.
We will design reductions to reduce the $m$-set $k$-packing,
$m$-dimensional $k$-matching and $t$-dominating set problems to the multilinear monomial testing problem,
and then use the above results to obtain our deterministic algorithms to solve those three problems.

\section{Preliminaries}\label{def-section}

For variables $x_1, \dots, x_n$,  for
$1\le i_1 < \cdots <i_k \le n$, $\pi =x_{i_1}^{s_1}\cdots
x_{i_t}^{s_t}$ is called a monomial. The degree of $\pi$, denoted
by $\mbox{deg}(\pi)$, is $\sum^t_{j=1}s_j$. $\pi$ is multilinear,
if $s_1 = \cdots = s_t = 1$, i.e., $\pi$ is linear in all its
variables $x_{i_1}, \dots, x_{i_t}$. For any given integer $q\ge
2$, $\pi$ is called a $q$-monomial if $1\le s_1, \dots, s_t \le q-1$.
In particular, a multilinear monomial is the same as  a $2$-monomial.

An arithmetic circuit, or circuit for short, is a directed acyclic
graph consisting of $+$ gates with unbounded fan-ins, $\times$ gates with two
fan-ins, and terminal nodes that correspond to variables. The size,
denoted by $s(n)$, of a circuit with $n$ variables is the number
of gates in that circuit. A circuit is  a formula
 if the fan-out of every gate is at most one, i.e.,
the underlying directed acyclic graph that excludes all the terminal nodes is a tree.
In other words, in a formula, only the terminal nodes can have more than one fan-out
(or out-going edge).

Given any multivariate polynomial represented by a circuit, for a fixed integer $q\ge 2$, 
the $q$-monomial testing problem asks whether 
is a $q$-monomial of degree $k$ in the sum-product expansion of the polynomial. 

Throughout this paper, the $O^*(\cdot)$ notation is used to
suppress $\mbox{poly}(n,k)$ factors in time complexity bounds.

By definition, any polynomial $F(x_1,\dots,x_n)$ can be expressed
as a sum of a list of monomials, called the sum-product expansion.
The degree of the polynomial is the largest degree of its
monomials in the expansion. With this expanded expression, it is trivial to
see whether $F(x_1,\dots,x_n)$ has a multilinear monomial, or a
monomial with any given pattern. Unfortunately, such an expanded expression is
essentially problematic and infeasible due to the fact that a
polynomial may often have exponentially many monomials in its
sum-product expansion. The challenge then is to test whether $F(x_1,\dots,x_n)$
has a multilinear monomial, or any other desired monomial, efficiently but
without expanding it into its sum-product representation.

For any integer $k \ge 1$, we consider the group
$Z^k_2$ with the multiplication $\cdot$ defined as follows. For
$k$-dimensional column vectors $\vec{x}, \vec{y} \in Z^k_2$ with
$\vec{x} = (x_1, \ldots, x_k)^T$ and $\vec{y} = (y_1, \ldots,
y_k)^T$, $\vec{x} \cdot \vec{y} = (x_1+y_1, \ldots, x_k+y_k)^T.$
$\vec{v}_0=(0, \ldots, 0)^T$ is the zero element in the group.
For any field $\mathcal{ F}$, the group algebra $\mathcal{F}[Z^k_2]$ is defined as
follows. Every element $u \in \mathcal{F}[Z^k_2]$ is a linear sum of the form
\begin{eqnarray}
 u &=& \sum_{\vec{x}_i\in Z^k_2,~ a_{i}\in \mathcal{F}} a_{i} \vec{x}_i. \nonumber
\end{eqnarray}
For any element
$v = \sum\limits_{\vec{x}_i\in Z^k_2,~ b_{i}\in \mathcal{F}} b_{i}
\vec{x}_i$,
we define
\begin{eqnarray}
 u + v  &=& \sum_{a_{i},~ b_{i}\in \mathcal{F},~  \vec{x}_i\in Z^k_2}  (a_i+b_i)
 \vec{x}_i, \  \mbox{and} \nonumber\\
u \cdot v &=& \sum_{a_i,~ b_j\in \mathcal{F},~ \mbox{ and }~\vec{x}_i,~ \vec{y}_j\in Z^k_2}  (a_i b_j)
(\vec{x}_i\cdot \vec{y}_j). \nonumber
\end{eqnarray}
For any scalar $c \in \mathcal{F}$,
\begin{eqnarray}
 c u &=& c \left(\sum_{\vec{x}_i\in Z^k_2, \ a_i\in \mathcal{F}} a_{i} \vec{x}_i\right)
 = \sum_{\vec{x}_i\in Z^k_2,\  a_{i}\in \mathcal{F}} (c a_{i})\vec{x}_i. \nonumber
\end{eqnarray}
The zero element in the group algebra $ \mathcal{ F}[Z^k_2]$ is
$ {\bf  0} = \sum_{\vec{v}} 0\vec{v}$, where $0$ is the zero element in $\mathcal{ F}$
and $\vec{v}$ is any vector in $\displaystyle Z_2^k$. For example,
${\bf  0} = 0\vec{v_0} = 0\vec{v}_1 + 0\vec{v}_2 + 0\vec{v}_3$,
for any $\displaystyle \vec{v}_i \in Z^k_2$, $1\le i\le 3$.
The identity element in the group algebra $\displaystyle \mathcal{ F}[Z^k_2]$ is
$ {\bf 1} = 1 \vec{v}_0 =  \vec{v}_0$, where $1$ is the identity element in $\mathcal{ F}$.
For any vector $\vec{v} =(v_1, \ldots, v_k)^T \in Z_2^k$, for
$i\ge 0$, let
$ (\vec{v})^i = (i v_1, \ldots, i v_k)^T.$
 In particular, when the field $\mathcal{ F}$ is  $Z_2$ (or in general, of characteristic $2$),
 in the group algebra $\mathcal{ F}[Z_2^k]$,
for any $\vec{z}\in Z_2^k$ we have $(\vec{v})^0 = (\vec{v})^2 =
\vec{v}_0$, and $\vec{z} + \vec{z} = \vec{0}$.

\section{Three Deterministic Algorithms} 

\subsection{The Packing Problem}

\begin{theorem}\label{packing-thm}
There is a deterministic algorithm for solving the $m$-set $k$-packing problem in time
$O^*(5.44^{mk})$.
\end{theorem}

\begin{proof}
Let $U$ be a set with $n$ elements. Given any collection  $\mathcal{ S}$  of $\mbox{\em poly}(n)$ many subsets  of $U$
such that all the members in $\mathcal{ S}$ have the same size $m\ge 3$, we need to decide
whether there are $k$ members in $\mathcal{ S}$ such that those members are pairwise disjoint.

We view each element in $U$ as a variable $x_i$.
For any member $A = \{x_{i_1},x_{i_2},\ldots,$ $x_{i_m}\} \in S$,
let $\pi(A) = x_{i_1} x_{i_2}\cdots x_{i_m}$ be a monomial of $m$ variables in $A$.
Define
\begin{eqnarray}\label{packing-expa}
F(\mathcal{ S}, k) &=& \left(~\sum_{A \in \mathcal{ S}} \pi(A)\right)^k.
\end{eqnarray}

Assume that there exists an $m$-set $k$-packing consisting of
$A_1, A_2, \ldots, A_k$ from $\mathcal{S}$, where $A_i$'s are mutually disjoint. Then,
$\pi = \pi(A_1)\pi(A_2)\cdots \pi(A_k)$ is a monomial in $F(\mathcal{S}, k)$.
Since $|A_i| = m$ and $A_i$'s are mutually disjoint, $\pi$ is multiplinear and has a degree of $mk$.

Now, suppose that $F(\mathcal{S}, k)$ has a degree $mk$ multilinear monomial
$\phi$. According to expression (\ref{packing-expa}),
$\phi = \phi_1\phi_2 \cdots \phi_k$ such that
$\phi_i = \pi(B_i)$ for some member $B_i \in \mathcal{S}$.
 Since  $\phi$ is multilinear, $\phi_i$ and $\phi_j$ have no common variables, for $i\not= j$.
 This means that $B_1,B_2,\ldots,B_k$ are pairwise disjoint.
 Recall that all the members in $\mathcal{S}$ have the same size $m$.
 Thus,  $B_1,B_2,\ldots,B_k$ form an $m$-set $k$-packing for $\mathcal{S}$.

It follows from the above analysis that, in order to find an $m$-set $k$-packing  for $\mathcal{S}$,
we only need to test whether there is  a  degree $mk$ multilinear monomial in
the sum-product expansion of $F(\mathcal{S}, k)$.
Since $F(\mathcal{ S}, k)$ can be easily represented by a formula with a $\mbox{\em poly}(n,k)$ size and
$2$-monomials are the same as multilinear monomials,
by Theorem \ref{thm-1}, this can be done deterministically
in time $O^*(5.44^{mk})$.
\end{proof}

\subsection{The Matching Problem}

\begin{theorem}\label{matching-thm}
There is a deterministic algorithm to solve the $m$-dimensional $k$-matching problem in time
$O^*(5.44^{(m-1)k})$.
\end{theorem}

\begin{proof}
Consider $U = U_1\times U_2\times \cdots \times U_m$ with mutully disjoint $U_i$'s.
For any given set $\mathcal{C} \subseteq U$ of $m$-component tuples, we need to decide
whether $\mathcal{C}$ contains a size $k$ subset $\mathcal{C'}$ such that
all the tuples in $\mathcal{C'}$ have no common components.

Following a reduction in \cite{koutis09}, we reduce the matching problem to
the multilinear monomial testing problem.
Let $U_i=\{u_{i1},u_{i2},\ldots,u_{in_i}\}$, $1\le i\le m$. For each element $u_{ij} \in U_i$, we design a variable
$x_{ij}$ to represent it. For any $m$-tuple $v=(u_{1j_1},u_{2j_2},\ldots,u_{mj_m})\in U$,
let $\pi(v) = x_{2j_2}\cdots x_{mj_m}$ be the monomial for $v$. Note that the variable $x_{1j_1}$ is intentionally
not used to construct $\pi(v)$. Obviously, $\pi(v)$ has degree $m-1$.
We partition $\mathcal{C}$ into $n_1$ pairwise disjoint
sets $\mathcal{C}_j$'s such that $\mathcal{C}_j$ is the set of all the tuples in $\mathcal{C}$ with the first component being
$u_{1j}$. Let $z$ be an additional {\em "marking"} variable.  Define
\begin{eqnarray} \label{matching-expa}
f(\mathcal{C}_j) & = & 1 + \sum_{v\in \mathcal{C}_j} z \pi(v),~~ 1\le j \le n_1, \\
F(\mathcal{C}, z) & =& f(\mathcal{C}_1) f(\mathcal{C}_2) \cdots f(\mathcal{C}_{n_1}). \label{matching-expb}
\end{eqnarray}

Assume that $\mathcal{C}$ has a size $k$ subset $\mathcal{M}$ with mutually disjoint tuples.
Let $\mathcal{M}_j = \mathcal{M} \cap \mathcal{C}_j$ and $t_j = |\mathcal{M}_j|$,  $1\le j\le n_1$. We have
$k = \sum_{j=1} t_j$.  Since tuples in $C_i$ share the same first component, 
we have $t_j=1$ or $0$, i.e., either $\mathcal{M}$ is empty or it has only one tuple. 
Let 
\begin{eqnarray} 
g(\mathcal{M}_j)  &=& 1 \mbox{ if }\mathcal{M}_j = \emptyset, \mbox{ or } \pi(v_{i_j}) \mbox{ if }
\mathcal{M}_j = \{v_{i_j}\}, \nonumber \\
G(\mathcal{M}) &=& g(\mathcal{M}_1) g(\mathcal{M}_2) \cdots g(\mathcal{M}_{n_1}). \nonumber 
\end{eqnarray}
Since tuples in $\mathcal{M}$ are mutually disjoint, it follows from
expression (\ref{matching-expa}) that $g(\mathcal{M}_j) = z^{t_j} \pi(\mathcal{M}_j)$, where
 $\pi(\mathcal{M}_j)$ is a degree $(m-1)t_j$ multilinear monomial of $x$-variables.
 Furthermore, by expression (\ref{matching-expb}),
 $G(\mathcal{M}) = \prod_{j=1}^{n_1} z^{t_j} \pi(\mathcal{M}_j)
 = z^k \prod_{j=1}^{n_1} \pi(\mathcal{M}_j)$, where $\prod_{j=1}^{n_1} \pi(\mathcal{M}_j)$  is multilinear
 and its degree  is $\sum_{j=1}^{n_1}[(m-1)t_j] = (m-1)\sum_{j=1}^{n_1}t_j = (m-1)k$.
 Let $\pi(\mathcal{M}) = \prod_{j=1}^{n_1} \pi(\mathcal{M}_j)$.
Thus, $F(\mathcal{C}, z)$ has a monomial $z^k\pi(\mathcal{M})$ such that
$\pi(\mathcal{M})$ is a degree $(m-1)k$ multilinear monomial of $x$-variables.

On the other hand, assume that $F(\mathcal{C},z)$ has, in its sum-product expansion, a monomial $z^k\phi$ such that
$\phi$ is a degree $(m-1)k$ multilinear monomial of $x$-variables. Then, by
expressions  (\ref{matching-expa}) and (\ref{matching-expb}),
 $\phi = \phi_1\phi_2 \cdots \phi_{n_1}$ with $\phi_j $
 being a monomial from $f(\mathcal{C}_{j})$ for $1\le j \le n_1$.
 Furthermore, by expression (\ref{matching-expa}), $\phi_j$ is either $1$ or
$z \pi(v_{i_j})$ for some tuple $v_{i_j} \in \mathcal{C}_j$. In the latter case,
the degree of $\phi_j$ is $1 + (m-1)$.
 Let $S$ be the set of all the $v_{i_j}$'s such that $\phi_j \not= 1$ and $\phi_j = z\pi(v_{i_j})$
 for some $v_{i_j} \in \mathcal{C}_j$. Let $\ell = |S|$.
 Then,
 $$
 z^k\phi = \prod_{v_{i_j}\in S} z\pi(v_{i_j}),
 $$
 and the degree of $z^t\phi$ is
 $$
 k + (m-1)k = (1 + (m-1))\ell = \ell + (m-1)\ell.
 $$
  This implies that $\ell = k$. Hence, it is easy to see that $S$ is an $m$-dimensional $k$-matching
  for $\mathcal{C}$.

Combining the above analysis, in order to decide whether there exists an $m$-dimensional $k$-matching for $\mathcal{C}$,
we only need to  test  whether there is  a monomial $z^k\pi$ in the sum-product expansion of $F(\mathcal{C}, z)$
such that $\pi$ is a degree $(m-1)k$ multilinear monomial of $x$-variables.
By expressions (\ref{matching-expa}) and (\ref{matching-expb}), 
$F(\mathcal{C},z)$ can be represented by a formula with a $\mbox{\em poly}(n, m)$ size, 
where $n=n_1+n_2+\cdots+n_m$. Recall that $2$-monomials are the same as multilinear monomials.
By Corollary \ref{cor}, the testing can be done deterministically
in  $O^*(5.44^{(m-1)k})$ time.
\end{proof}

\subsection{The $t$-Dominating Set Problem}

\begin{theorem}\label{dominating-thm}
There is a deterministic algorithm for solving the $t$-dominating set problem in time $O^*(5.44^t)$.
\end{theorem}

\begin{proof}

Given any simple undirected graph $G=(V,E)$,
we need to find a minimal $k$ such that there exists a size $k$ set $S\subseteq V$
that dominates at least $t$ distinct nodes.
We first consider the decision version of this problem: For any two given parameters $k$ and $t$,
we decide whether there is a set
set $S\subseteq V$ such that $S$ has a size of at most $k$ and dominates at least $t$ distinct nodes.
Like in \cite{koutis09}, we will reduce this decision version of the $t$-dominating set problem
to the multilinear monomial testing problem.

Let $V = \{v_1,v_2,\ldots,v_n\}$. For each node $v_i$ we design a variable $x_i$ to represent it.
We also use $N(v_i)$ to denote the set of nodes adjacent to $v_i$, i.e.,
$N(v_i) = \{v_j| (v_i, v_j) \in E\}$. Since $G$ is undirected, the edge $(v_i, v_j) = (v_j, v_i)$.
Also, since $G$ is simple, $v_i\not\in N(v_i)$.
Let $z$ be an additional {\em "marking"} variable to mark the desired monomials
that will be given in later discussions. For any given integer $k\ge 1$, define
\begin{eqnarray}\label{dominating-expa}
&& f(v_i)  =  (1 + zx_i) \prod_{v_j \in N(v_i)}(1+z x_j), \mbox{ for any } v_i\in V,\\
&& F_k(x_1,x_2,\ldots,x_n,z)  =  \left(\sum_{v_i\in V} f(v_i)\right)^k,~~ k\ge 1. \label{dominating-expb}
\end{eqnarray}

Assume that there is a set $S\subseteq V$ with size $k$ that dominates at least $t$ nodes. Let
$D = S \cup \{N(v_i)|v_i\in S\}$. Then, $D$ is dominated by $S$, hence $|D| \ge t$.
We select a set $W = \{x_{j_1},x_{j_2},\ldots,x_{j_t}\}$ of $t$ distinct nodes from $D$.
Since $W$ is dominated by $S$, we can partition $W$ into $s\le k$ mutually disjoint sets
$W_1,W_2,\ldots,W_s$ that are dominated by $s$~ distinct nodes $v_{\ell_1},v_{\ell_2},\ldots,v_{\ell_s}$ in $S$, respectively.
Let $S' = \{ v_{\ell_1},v_{\ell_2},\ldots,v_{\ell_s}\}$.  Let
\begin{eqnarray}
g(v_i) &=& \left\{
\begin{array}{ll}
1, & \mbox{ if } v_i \not\in S', \\
\prod_{x_{j_\iota} \in W_{\tau}} z x_{j_\iota}, & \mbox{ if } v_i = v_{\ell_\tau} \in S'.
\end{array}
\right. \nonumber \\
g(S) & = & \prod_{v_i\in S} g(v_i). \nonumber
\end{eqnarray}
Then,
\begin{eqnarray}
g(S) & = & z x_{j_1} z x_{j_2} \cdots z x_{j_t} = z^t x_{j_1} x_{j_2} \cdots x_{j_t}. \nonumber
\end{eqnarray}
Furthermore, from expressions (\ref{dominating-expa}) and (\ref{dominating-expb}),
$g(v_i)$ is a monomial in the sum-product expansion of $f(v_i)$, 
so is $g(S)$ in the sum-product expansion of $F_k$. Because
$x_{j_1}, x_{j_2}, \ldots x_{j_t}$ are distinct, their product is a degree $t$ multilinear monomial.
Hence, $F_k$ has a monomial $z^t \pi$ in its sum-product expansion with $\pi$ being a degree $t$ multilinear
monomial of $x$-variables.

Now, assume that there is a monomial $\phi = z^t\pi$ in the sum-product expansion of $F_k$
such that $\pi$ is mutilinear with degree $t$. According to expressions  (\ref{dominating-expa}) and (\ref{dominating-expb}),
 $\phi = \phi_1\phi_2 \cdots \phi_k$ with $\phi_i = z^{t_i}\pi_i $
 being a monomial from $f(v_{i_j})$ for distinct $k$ nodes $v_{i_1},v_{i_2},\ldots, v_{i_k}$.
 Let $S' = \{v_{i_j} | \phi_i \not= 1\}$. For any $v_{i_j} \in S'$,
 let $N'(v_{i_j}) = \{x_{\ell_\tau} |   x_{\ell_\tau} \mbox{ is in } \pi_i \}$.
 Since each $z$-variables is paired with a $x$-variable in $f(v_{i_j})$ according to expression (\ref{dominating-expa}),
 the degree of $\phi_i$ is $t_i = |N'(v_{i_j})|$.
 Also, by expression (\ref{dominating-expa}), nodes in $N'(v_{i_j})$ are dominated by $v_{i_j} \in S'$.
 Since $\pi$ is multilinear,
 the total number of nodes dominated by nodes in $S'$ is at least
 $$
 |\{N'(v_{i_j}) | v_{i_j}\in S'\}| = t_1 + t_2\cdots + t_k = t.
 $$
 Therefore, there is a set of nodes with size at most $k$ to dominate at least $t$ nodes.

Following the above analysis, in order to decide whether there exists a node set of size  $k$
that dominates at least $t$ nodes,
we only need to test whether $F_k$ has a monomial $z^t\pi$ such that $\pi$ is
a degree $t$ multilinear monomial of $x$-variables in
the sum-product expansion of $F_k$.  By expressions (\ref{dominating-expa}) and (\ref{dominating-expb}),
$F_k(x_1,x_2,\ldots,x_n,z)$ can be easily represented by a formula of $\mbox{\em poly}(n,k)$ size.
Because $2$-monomials are the same as multilinear mnonomials,
by Corollary \ref{cor}, this can be done deterministically in time $O^*(5.44^t)$.

Finally, we can use the algorithm for the decision version of the $t$-dominating set problem
to find the minimal $k$ as follows: For $k=1,2,\ldots,t$, run the algorithm. If it says {\em "yes"}
for the first $k$, then stop and return this $k$. With those additional efforts, the total time complexity
bound remains as $O^*(5.44^t)$.
\end{proof}

\section{Faster Multiplication of Group Algebraic Elements}

For $0\le i\le 2^k -1$,  let $\vec{v}_i$ denotes the vector in $Z_2^k$ that corresponds to the $k$-bit binary representation of $i$.
E.g., when $k=4$, $\vec{v}_2 = (0,0,1,0)^T$ and $\vec{v}_9 = (1,0,0,1)^T$.
Recall from Section \ref{def-section} that $\vec{v}_i \cdot \vec{v}_j$ is the 
component-wise addition $(\mbox{\em mod}~ 2)$ of the two vectors. We use $\vec{v}_i \times \vec{v}_j$ to denote the inner product of
 $\vec{v}_i$ and $\vec{v}_j$. 

Given two elements $\vec{x}$ and $\vec{y}$ in the group algebra $\mathcal{  F}(Z^k_2)$, let
$\vec{x} = \sum_{i=0}^{2^k - 1} a_i \vec{v}_i$ and $\vec{y} = \sum_{i=0}^{2^k - 1} b_i \vec{v}_i.$
Also, $\vec{x}$ and $\vec{y}$ can be represented as $2^k$-dimensional vectors
$A(\vec{x})$ and $A(\vec{y})$ such that $A(\vec{x}) = (a_0,a_1,\ldots,a_{2^k -1})^T$ and
$A(\vec{y}) = (b_0,b_1,\ldots,b_{2^k -1})^T$.
 It is easy to know that
\begin{eqnarray}\label{xy}
\vec{x}\vec{y} &=& \sum_{t=0}^{2^k -1} \left(~\sum_{\vec{v}_i \cdot \vec{v}_j = \vec{v}_t} a_i b_j~\right) \vec{v}_{t}.
\end{eqnarray}

Naively, the time needed to calculate the above multiplication is $O(4^k \log^2|\mathcal{  F}|)$.
Unfortunately, this is inefficient for designing our monomial testing algorithm in Section \ref{FDTM-section}.
Williams \cite{williams09} briefly mentioned that Walsh-Hadamard transformation can be used
to speed up the multiplication to $O(k2^k \log^2|\mathcal{  F}|)$.
Although discrete Fourier transformations are commonly used to multiply two polynomials,
yet the multiplication of two group algebraic elements is substantially different from that of two polynomials.
Hereby, it is necessary to give a formal justification about why Walsh-Hadamard transformation
can be used for computing $\vec{x}\vec{y}$  as in expression (\ref{xy}).

Let $H_k$ denote the $2^k\times 2^k$ Walsh-Hadamard matrix. Here, we do not consider matrix normalization,
so entries in $H_k$ are either $1$ or $-1$. It is known that the
entry of $H_k$ at row $i$ and column $j$ is
\begin{eqnarray}\label{wh-exp1}
h_{ij} & = & (-1)^{\vec{v}_i \times \vec{v}_j}.
\end{eqnarray}

\begin{lemma}\label{wh-lem1}
for any given $0\le j \le 2^k-1$, we have
\begin{eqnarray}\label{wh-exp2}
\sum_{i=0}^{2^k-1} (-1)^{\vec{v}_i \times \vec{v}_j} &=&
\left\{\begin{array}{ll}
2^k, & \mbox{ if $j=0$,}\\
0, & \mbox{ if $j\not= 0$}.
\end{array}
\right.
\end{eqnarray}
\end{lemma}

\begin{proof}
According to (\ref{wh-exp1}),
$\sum_{i=0}^{2^k-1} (-1)^{\vec{v}_i \times \vec{v}_j} = \sum_{i=0}h_{ij}.$
The right part is the sum of all the entries at column $j$ of $H_k$, which is $2^{k}$ for the first column when
$j=0$ or $0$ for any other column when $j\not= 0$.
\end{proof}

We perform the following Walsh-Hadamard transformation:
\begin{eqnarray}
\vec{X} & = & H_k~ A(\vec{x}), \nonumber \\
\vec{Y} & = & H_k~ A(\vec{y}), \nonumber \\
\vec{Z} &=& \vec{X} \times \vec{Y},  \nonumber \\
\vec{W} &=& \frac{1}{2^k}~ H_k~ \vec{Z}. \nonumber
\end{eqnarray}

\begin{lemma}\label{wh-lem2}
Let $\vec{W} = (W_0,W_1,\ldots,W_{2^k-1})^T$. The above transformation yields
\begin{eqnarray}\label{wh-lem-exp}
\vec{x}\vec{y} = \sum_{t=0}^{2^k-1} W_t \vec{v}_t, \mbox{  i.e., }  W_t = \sum_{\vec{v}_i \cdot \vec{v}_j = \vec{v}_t} a_i b_j,
\end{eqnarray}
and the time needed to complete the transformation is $O^*(k2^k \log^2|\mathcal{  F}|)$.
\end{lemma}

\begin{proof}
Let the $i$-th component of $\vec{X}$, $\vec{Y}$ and $\vec{Z}$ be $X_i$, $Y_i$ and $Z_i$, respectively.
For $0\le t\le 2^k -1$, according to the above transformation,
\begin{eqnarray}\label{wh-exp3}
W_t & = & \frac{1}{2^k} \sum_{\ell =0}^{2^k-1} h_{t\ell}Z_{\ell}
     =  \frac{1}{2^k} \sum_{\ell =0}^{2^k-1} h_{t\ell}X_{\ell}Y_{\ell} \nonumber \\
     &=& \frac{1}{2^k} \sum_{\ell =0}^{2^k-1} h_{t\ell} \left(\sum_{i =0}^{2^k-1} h_{\ell i}a_i \right)\left(\sum_{j =0}^{2^k-1} h_{\ell j} b_j\right) \nonumber \\
    & = & \frac{1}{2^k} \sum_{i =0}^{2^k-1} \sum_{j =0}^{2^k-1} a_i b_j \left(\sum_{\ell =0} ^{2^k -1} h_{t\ell}h_{\ell i}h_{\ell j}\right)
\end{eqnarray}

By (\ref{wh-exp1}) and (\ref{wh-exp2}),
\begin{eqnarray}\label{wh-exp4}
\sum_{\ell =0} ^{2^k -1} h_{t\ell}h_{\ell i}h_{\ell j}
& = &  \sum_{\ell =0} ^{2^k -1}
(-1)^{\vec{v}_t \times \vec{v}_{\ell} + \vec{v}_{\ell} \times \vec{v}_{i} + \vec{v}_j \times \vec{v}_{\ell}} \nonumber \\
&= &  \sum_{\ell =0}^{2^k -1} (-1)^{\vec{v}_{\ell} \times (\vec{v}_t  \cdot \vec{v}_i \cdot \vec{v}_j)} \nonumber \\
&=&\left\{
\begin{array}{ll}
2^k, & \mbox{ if $\vec{v}_t  \cdot \vec{v}_i \cdot \vec{v}_j = \vec{v}_0$,}\\
0, & \mbox{ otherwise}.
\end{array}
\right.
\end{eqnarray}

Note that $\vec{v}_t  \cdot \vec{v}_i \cdot \vec{v}_j = \vec{v}_0$ if and only of
$\vec{v}_t = \vec{v}_i \cdot \vec{v}_j$. Hence, by (\ref{wh-exp3}) and (\ref{wh-exp4}),
\begin{eqnarray}\label{wh-exp5}
W_t & = & \sum_{\vec{v}_i+\vec{v}_j = \vec{v}_t} a_i b_j.
\end{eqnarray}
Therefore, (\ref{wh-lem-exp}) follow from  (\ref{wh-exp5}).

It is known from the property of Walsh-Hadamard transformation \cite{maslen95}
that $\vec{X}$, $\vec{Y}$ and $\vec{W}$ can be computed in $O(k2^k \log^2 |\mathcal{  F}|)$ time.
Computing $\vec{Z}$ takes $O(2^k \log^2 |\mathcal{  F}|)$. Thus, the total time to compute $\vec{W}$ (hence $\vec{x}\vec{y}$)
is $O(k2^k \log^2 |\mathcal{  F}|)$.
\end{proof}

\section{A Special Case of Read-Once Formulas}

Recently, many efforts have been made on the study of polynomial identity testing for formulas,
in particular read-once formulas (e.g., \cite{shpilka08,raz05}).
Here, we consider a special case of read-once formulas, called $S$-read-once formulas.
Given a formula $\mathcal{ C}$, we say $\mathcal{ C}$ is $S$-read-once,
if it is read-once (i.e., each variable appears once in the circuit), 
its underlying graph of $\mathcal{C}$ including all the terminal nodes is a tree, and
 each terminal node in it is connected to a $\times$ gate $g$ such that $g$ have one value input $\alpha$
 from a field $\mathcal{ F}$ and the other variable input $x$.
The output of $g$ is $\alpha x$. We need to deal with $S$-read-once formula in Section \ref{FDTM-section}.
A simple and direct algorithm for polynomial identity testing for this type of formulas will help us
design our new monomial testing algorithm.

\begin{lemma}\label{pit-lem}
For any $n$-variate polynomial $F(x_1,x_1,\ldots,x_n)$ that is represented by an $S$-read-once formula
$\mathcal{  C}$ with size $s(n)$, we can test whether $F$ is identically zero or not in time $O(s^4(n)log |\mathcal{F}|)$.
Here, the coefficients of $F$ are in a field $\mathcal{F}$.
\end{lemma}

\begin{proof}
Our approach is motivated by the polynomial identity testing algorithm of
Raz and Shpilka \cite{raz05} for pure circuits. We shall reconstruct $\mathcal{C}$
by reducing its nodes in the following two steps:

{\bf Step 1.} We start with the nodes that are parent nodes of terminal nodes.
Since $\mathcal{  C}$ is $S$-read-once, each of such a node $g$ is a $\times$ gate with two input nodes representing
a value input $\alpha \in \mathcal{  F}$ and a variable input $x$. We delete the two input nodes for $g$ and replace $g$ with
a leaf node that represent $\alpha x$. We complete this type of deletion and replacement for all such nodes $g$'s.
Obviously the resulting circuit, denoted as $\mathcal{  C}'$,   is equivalent to the original circuit $\mathcal{  C}$.

{\bf Step 2.} Now, we consider each node $g$ in $\mathcal{  C'}$ that has only leaf nodes as its children.
If $g$ is a $+$ gate with input leaf nodes $g_1, g_2, \ldots, g_t$, then we delete $g_1, g_2, \ldots, g_t$ and replace
$g$ with $g_1 + g_2 + \cdots g_t$.
Note that $g_i$ is a linear sum of variables.  Hence, $g$ is replaced by a new linear sum of variables.
Again, it is easy to see that the resulting circuit, still denoted as $\mathcal{  C}'$, is equivalent to the original circuit.

When $g$ is a $\times$ gate with two input leaf nodes $g_1$ and $g_2$, we know that
$g_i = \sum_{j=1}^{n_i} \alpha_{i_j} x_{i_j}$, for $i=1,2$. The output $f(g)$ of $g$ is
\begin{eqnarray}\label{pit-expa}
f(g) & = & \sum_{j=1}^{n_1} \alpha_{1_j} x_{1_j} \sum_{\ell=1}^{n_2} \alpha_{2_\ell} x_{2_\ell} \nonumber \\
     & = & \sum_{j=1}^{n_1} \sum_{\ell=1}^{n_2} \alpha_{1_j}\alpha_{2_\ell} x_{1_j}x_{2_\ell}
\end{eqnarray}

Also, in the circuit $\mathcal{  C'}$, mark the path $L$ from $g$ to the root gate. Starting at the parent of $g$,
find the first $+$ gate along the path $L$ and denoted it  as $M$. Let $w_1, w_2, \ldots, w_{t}$ be the consecutive $\times$ gates
from the parent of $g$ to the gate $M$ on the path $L$. Let $u_i$ be the other input to $w_i$ that is not on $L$.
Let $F_1$ be the polynomial computed by the circuit obtained from $\mathcal{  C'}$ by deleting the gates $g$,
$w_1$, $w_2, \ldots, w_t$. We continue to find more $\times$ gates after $M$ on $L$. Let $h_1, \ldots, h_m$ be the list of
all the $\times$ gates after $M$ on $L$. Let $o_i$ be the other input to $h_i$ that is not on $L$.
Let $F_2 = u_1 u_2 \cdots u_t o_1 \cdots o_m$.
Since $\mathcal{  C}$ (and hence $\mathcal{  C'}$) is $S$-read-once, we have by expression (\ref{pit-expa})
\begin{eqnarray}\label{pit-expb}
F &\equiv& F_1 + f(g) F_2 \nonumber \\
  & = & F_1 +  \sum_{j=1}^{n_1} \sum_{\ell=1}^{n_2} x_{1_j}x_{2_\ell} (\alpha_{1_j}\alpha_{2_\ell} F_2),
\end{eqnarray}
and $F_1, F_2$ and $f(g)$ do not have any common variables. Therefore, by expression (\ref{pit-expb}), we have
\begin{eqnarray}\label{pit-expc}
F \equiv 0
\Longleftrightarrow F_1 = 0 \mbox{ and } \alpha_{1_j}\alpha_{2_\ell} F_2 = 0, 1\le j\le n_1, 1\le \ell \le n_2.
\end{eqnarray}
Let $d$ be the greatest common divisor of $\alpha_{1_j}\alpha_{2_\ell},  1\le j\le n_1$, $1\le \ell \le n_2$. Then,
\begin{eqnarray}\label{pit-expd}
\alpha_{1_j}\alpha_{2_\ell} F_2  =  0,  1\le j\le n_1, 1\le \ell \le n_2
&\Longleftrightarrow & d F_2 = 0.
\end{eqnarray}
Hereby, by  (\ref{pit-expc}) and (\ref{pit-expd}),
\begin{eqnarray}\label{pit-expe}
F \equiv 0 & \Longleftrightarrow & F_1 = 0 \mbox{ and } d F_2 = 0.
\end{eqnarray}
We choose a brand new variable $x$, delete $g_1$ and $g_1$ and replace $g$ with $dx$ in $\mathcal{  C'}$. By (\ref{pit-expe}), we have
\begin{eqnarray}\label{pit-expf}
F \equiv 0 & \Longleftrightarrow & \mathcal{  C'} \equiv 0.
\end{eqnarray}

We repeat Step 2 to continue reducing nodes in $\mathcal{  C'}$ until $\mathcal{C'}$ has only the root node.
At that point, $\mathcal{  C'}$ represents a linear sum of variables so that we can easily check whether the sum is identical to
zero or not.

The above process needs to perform at most $s(n)$ reduction operations  for every gate (or node) in $\mathcal{  C}$.
For a $+$ gate, the related reduction is done in time $O(s(n))$.
For a $\times$  gate $g$ with two inputs $g_1$ and $g_2$, as in (\ref{pit-expa}) and (\ref{pit-expd}),
we need to compute $\alpha_{1_j}\alpha_{2_\ell}$ for $1\le j\le n_1$  and  $1\le \ell \le n_2$ and their greatest common divisor
$d$. Since $n_1\le s(n)$ and $n_2\le s(n)$, the time needed is $O(s^3(n)\log |\mathcal{F}|)$. By adding the time needed for every gate,
the total time is $O(s^4\log |\mathcal{F}|)$.
\end{proof}

\section{Circuit Reconstruction and Variable Replacements}\label{trans}

In this section, we shall introduce the circuit reconstruction and variable replacement techniques
that are developed in \cite{bill13} to transform the $q$-monomial testing problem to the multilinear
monomial testing problem. This transformation is an extension of the method designed in
\cite{bill12} for formula to general circuits.

For any given polynomial $F(x_1,x_2,\ldots,x_n)$ represented by a circuit $\mathcal{ C}$
of size $s(n)$, we reconstruct the circuit $\mathcal{ C}$ and replace variables in three
steps as follows:



{\bf Duplicating terminal nodes.} For each variable $x_i$,
if $x_i$ is the input to a list of gates $g_1, g_2, \ldots, g_{\ell}$, then create $\ell$
terminal nodes $u_1, u_2, \ldots, u_{\ell}$ such that each of them represents
a copy of the variable $x_i$ and $g_j$ receives input from $u_j$, $1\le j\le \ell$.

Let $\mathcal{ C}^*$ denote the reconstructed circuit after the above step.
Obviously, both circuits $\mathcal{ C}$ and $\mathcal{ C^*}$ compute the same polynomial $F$.

{\bf Adding  new $\times$ gates and new variables}.
For every edge $e_i$ in $\mathcal{ C^*}$ (including every edge between a gate and a terminal node)
such that $e_i$ conveys the output of $u_i$ to $v_i$,
add a new $\times$ gate $g_i$ that multiplies the output of $u_i$ with a new variable $z_i$ and passes
the outcome to $v_i$.

Assume that a list of $h$ new $z$-variables
$z_1, z_2, \ldots, z_h$ have been introduced into the circuit $\mathcal{ C'}$.
Let $F'(z_1, z_2, \ldots, z_h, x_1, x_2,
\ldots, x_n)$ be the new polynomial represented by $\mathcal{ C'}$.

{\bf Variable replacements:} Here, we start with the new circuit $\mathcal{ C'}$
that computes $F'(z_1, z_2, \ldots, z_h, x_1, x_2,\ldots, x_n)$.
For each variable $x_i$, we replace it with a "weighted" linear sum  of $q-1$  new $y$-variables
$y_{i1},y_{i2},\ldots,y_{i(q-1)}$.
The replacements work as follows:
For each variable $x_i$, we first
add $q-1$ new terminal nodes that represent  $q-1$ many $y$-variables $y_{i1},y_{i2},\ldots,y_{i(q-1)}$.
Then, for each terminal node $u_j$ representing $x_i$ in $\mathcal{ C'}$, we replace $u_j$ with a $+$ gate.
Later, for each new $+$ gate $g_j$ that is created for $u_j$ of $x_i$, let $g_j$ receive input from $y_{i1},y_{i2},\ldots,y_{i(q-1)}$.
That is, we add an edge from each of such $y$-variables to $g_j$. Finally, for each edge $e_{ij}$ from $y_{ij}$ to $g_j$,
replace $e_{ij}$ by a new $\times$ gate that takes inputs from  $y_{ij}$ and a new $z$-variable $z_{ij}$ and sends the output to $g_j$.

Let $\mathcal{ C''}$ be the circuit resulted from the above transformation, and
$$
G(z_1,\ldots,z_h,y_{11},\ldots,y_{1(q-1)},\ldots,y_{n1},\ldots,y_{n(q-1)})
$$
be the polynomial computed by the circuit $\mathcal{ C''}$.
The following two lemmas are obtained in \cite{bill13}.

\begin{lemma}(\cite{bill13})\label{rtm-lem1}
Let the $t$ be the length of longest path from the root gate of $\mathcal{ C}$ to its terminal nodes.
$F(x_1,x_2,\ldots,x_n)$ has a monomial $\pi$ of degree $k$  in its sum-product expansion if and only if there is
a monomial $\alpha \pi$ in the sum-product expansion of $F'(z_1, z_2, \ldots, z_h, x_1, x_2,\ldots, x_n)$
such that $\alpha$ is a multilinear monomial of $z$-variables with degree $\le tk + 1$.
Furthermore, if $\pi$ occurs more than once in the sum-product expansion of $F'$, then
every occurrence of $\pi$ in $F'$ has a unique coefficient $\alpha$; and any two different monomials of $x$-variables
in $F'$ will have different coefficients that are multilinear products of $z$-variables.
\end{lemma}

\begin{lemma}(\cite{bill13})\label{rtm-lem2}
Let $F(x_1,x_2,\ldots,x_n)$ be any given polynomial represented by
a circuit $\mathcal{ C}$ and $t$ be the length of the longest path of $\mathcal{ C}$. For any fixed integer $q\ge 2$,
$F$ has a $q$-monomial of $x$-variables with degree $k$, then
$G$ has a unique multilinear monomial $\alpha \pi$ such that $\pi$ is a degree $k$ multilinear monomial of $y$-variables
and $\alpha$ is a multilinear monomial of $z$-variables with degree $\le k(t+1) +1 $. If $F$ has no $q$-monomials, then
$G$ has no multilinear monomials of $y$-variables, i.e., $G$ has no monomials of the format $\beta \phi$ such that
$\beta$ is a monomial of $z$-variables and
 $\phi$ is a multilinear monomial of $y$-variables.
\end{lemma}

\section{A Faster Deterministic Algorithm for Testing $q$-Monomials}\label{FDTM-section}

Chen \cite{bill13} obtained a $O^*(12.8^k)$ deterministic algorithm for testing
$q$-monomials in a multivariate polynomial represented by a formula. In section,
we shall devise a faster algorithm with a $O^*(5.44^k)$ time bound.
In contrast to the two derandomization processes in \cite{bill13,chen12b},
we first the near optimal family of perfect hashing functions by Naor {\em et al.} \cite{naor95}
to derandomize the group algebraic variable replacements  and Lemma \ref{pit-lem} to derandomize
the polynomial identity testing for $S$-read-once formulas. In addition,
we use Walsh-Hadamard transformation to speed up the multiplication of
group algebraic elements, which is formally justified by Lemma \ref{wh-lem2}.

\begin{definition}\label{def-hash}
(See, Chen {\em et al.} \cite{jianer-chen07}, Naor {\em et al.} \cite{naor95}) Let $n$ and $k$ be two integers such that $1\le k\le n$. Let
$\mathcal{ A} =\{1, 2, \ldots, n\}$ and $\mathcal{ K} = \{1, 2, \ldots,
k\}$. For any family $\mathcal{H}$ of functions mapping from $\mathcal{A}$ to $\mathcal{K}$, we say
$\mathcal{H}$ is an $(n,k)$-family of {\em perfect
hashing functions} if for any subset $S$ of $k$ elements in $\mathcal{
A}$, there is an $h \in \mathcal{ H}$ that is injective
from $S$ to $\mathcal{ K}$, i.e., for any $x, y \in S$, $h(x)$ and
$h(y)$ are distinct elements in $\mathcal{ K}$.
\end{definition}

We assume, without loss of generality, that
when a polynomial has  $q$-monomials in its sum-product expansion,
one of the $q$-monomials has exactly a  degree of $k$ and
all the rest of those will have degrees at least $k$.

Consider any given polynomial $F(x_1,x_2,\ldots,x_n)$ that is represented by a formula $\mathcal{C}$ of size
$s(n)$. Let $d = \log (k(s(n)+1)+1) + 1$ and $\mathcal{ F} = \mbox{GF}(2^d)$ be a finite field of $2^d$ elements.
Note that $\mathcal{ F}$ has characteristic $2$.

\begin{quote}
Algorithm \mbox{FDTM} (\underline{F}aster \underline{D}eterministic
\underline{T}esting of $q$-\underline{M}onomials):
\begin{description}
\item[1.] Following Section \ref{trans}, reconstruct $\mathcal{ C}$ to obtain $\mathcal{ C}^*$
 that computes the same polynomial $F$ and then introduce new $z$-variables to $\mathcal{ C}^*$
to obtain $\mathcal{ C'}$ that computes $F'(z_1,z_2,\ldots,z_h,x_1,x_2,\ldots, x_n)$.
Perform variable replacements to obtain $\mathcal{ C''}$ that transforms $F'$ to
$$
G(z_1,\ldots,z_h,y_{11},\ldots,y_{1(q-1)},\ldots,y_{n1},\ldots,y_{n(q-1)}).
$$

\item[2.] Construct with the algorithm by Naor {\em at el.}
\cite{naor95}  a $((q-1)n s(n), k)$-family of perfect hashing functions
$\mathcal{ H}$ of size $e^k k^{O(\log k)}\log^2 ((q-1)n s(n))$.

\item[3.] Select $k$ linearly
independent vectors $\vec{v}_1,\ldots,\vec{v}_{k} \in Z^{k}_2$. (No
randomization is needed at this step, either.)

\item[4] For each perfect
hashing function $\lambda \in\mathcal{ H}$ do
\begin{description}
\item[4.1.] Let $\tau(i,j)$
be any given one-to-one mapping from $\{(i,j) | 1\le i\le n \mbox{\ and\ } 1\le j\le q-1\}$
to $\{1,2,\ldots,(q-1)n\}$ to label variables $y_{ij}$.
Replace each variable $y_{ij}$ in $G$ with
$(\vec{v}_{\lambda(\tau(i,j))} + \vec{v}_0)$, $1\le i \le n$ and $1\le j\le q-1$.

\item[4.2.] Use $\mathcal{ C''}$ to calculate
\begin{eqnarray}\label{exp-thm-dt}
G'&=& G(z_1,\ldots,z_h,(\vec{v}_{\lambda(\tau(1,1))}+\vec{v}_0),\ldots,(\vec{v}_{\lambda(\tau(1,q-1))}+\vec{v}_0),
 \nonumber \\
&& \hspace{7mm}
\ldots,(\vec{v}_{\lambda(\tau(n,1))}+\vec{v}_0),\ldots,(\vec{v}_{\lambda(\tau(n,q-1))}+\vec{v}_0)) \nonumber \\
&=&  \sum_{j=1}^{2^k} f_j(z_1,\ldots,z_h) \cdot \vec{v}_j,
\end{eqnarray}
where each $f_j$ is a polynomial of degree $\le k(s(n)+1)+1$ (see, Lemma \ref{rtm-lem2}) over the finite
field $\mathcal{ F}=\mbox{GF}(2^d)$, and $\vec{v}_j$ with $1\le j\le 2^k$ are the $2^k$ distinct vectors in
$Z^k_2$.

\item[4.3.] Perform polynomial identity testing with the algorithm of Lemma \ref{pit-lem}
for every $f_j$ over $\mathcal{ F}$. Stop
and return {\em "yes"} if one of them is not identical to zero.
\end{description}

\item[5.] If all perfect hashing functions  $\lambda \in \mathcal{ H}$  have
been tried without returning {\em "yes"}, then stop and output {\em "no"}.
\end{description}
\end{quote}

The group algebra technique established by Koutis \cite{koutis08} assures
the following two properties:

\begin{lemma}\label{rtm-lem3}
(\cite{koutis08})~  Replacing all the variables $y_{ij}$
in $G$ with group algebraic elements $\vec{v}_{ij}+\vec{v}_0$ will make all
monomials $\alpha \pi$ in $G'$ to become zero, if $\pi$ is non-multilinear with respect to $y$-variables. Here,
$\alpha$ is a product of $z$-variables.
\end{lemma}

\begin{proof}
Recall that $\mathcal{ F}$ has characteristic $2$.
For any $\vec{v} \in Z^k_2$, in the group algebra $\mathcal{ F}[Z^k_2]$,
\begin{eqnarray}\label{rt-1}
(\vec{v}+\vec{v}_0)^2 &=& \vec{v}\cdot\vec{v} + 2\cdot\vec{v}\cdot\vec{v}_0 + \vec{v}_0\cdot\vec{v}_0 \nonumber \\
 &=&\vec{v}_0+ 2\cdot\vec{v} + \vec{v}_0 \nonumber \\
 &=& 2\cdot \vec{v}_0 + 2\cdot\vec{v} = {\bf 0}.
\end{eqnarray}
Thus, the lemma follows directly from expression (\ref{rt-1}).
\end{proof}

\begin{lemma}\label{rtm-lem4}
(\cite{koutis08})~ Replacing all the variables $y_{ij}$ in $G$
with group algebraic elements $\vec{v}_{ij}+\vec{v}_0$ will make any monomial $\alpha \pi$
to become zero,  if and only if  the vectors $\vec{v}_{ij}$  are linearly dependent in the vector space $Z^k_2$.
Here, $\pi$ is a multilinear monomial of $y$-variables and $\alpha$ is a product of $z$-variables,
Moreover, when $\pi$ becomes non-zero after the replacements,
it will become the sum of all the vectors in the linear space spanned by those vectors.
\end{lemma}

\begin{proof}
The analysis below gives a proof for this lemma.
Suppose $V$  is a set of  linearly dependent vectors
in $Z^k_2$. Then, there exists a nonempty subset $T \subseteq V$ such that $\prod_{\vec{v}\in T} \vec{v}= \vec{v}_0$.
For any $S\subseteq T$, since
$\prod_{\vec{v}\in T} \vec{v} =
(\prod_{\vec{v}\in S} \vec{v}) \cdot (\prod_{\vec{v}\in T-S} \vec{v})$, we have
$\prod_{\vec{v}\in S} \vec{v} = \prod_{\vec{v}\in T-S} \vec{v}$.
Thereby, we have
\begin{eqnarray}
\prod_{\vec{v}\in T }(\vec{v} + \vec{v}_0)
&=& \sum_{S\subseteq T} (\prod_{\vec{v}\in S}\vec{v}) = {\bf 0}, \nonumber
\end{eqnarray}
since every $\prod_{\vec{v}\in S}\vec{v}$ is paired by the same
$\prod_{\vec{v}\in T-S}\vec{v}$ in the sum above and the addition of the pair
is annihilated because $\mathcal{ F} $ has characteristic $2$. Therefore,
\begin{eqnarray}
\prod_{\vec{v}\in V}(\vec{v} + \vec{v}_0)
&=& \left(~\prod_{\vec{v}\in T} (\vec{v}+ \vec{v}_0)\right) \cdot \left(~\prod_{\vec{v}\in V-T}(\vec{v}+ \vec{v}_0)\right)  \nonumber \\
& =&  0 \cdot \left(~\prod_{\vec{v}\in V-T}(\vec{v}+ \vec{v}_0)\right) = {\bf 0}. \nonumber
\end{eqnarray}

Now consider that vectors in $V$ are linearly independent.
For any two distinct subsets $S, T\subseteq V$, we must have
$\prod_{\vec{v}\in T} \vec{v} \not=\prod_{\vec{v}\in S} \vec{v}$, because otherwise
vectors in $S \cup T - (S \cap T)$ are linearly dependent, implying that vectors in $V$ are linearly dependent.
Therefore,
\begin{eqnarray}
\prod_{\vec{v}\in V}(\vec{v} + \vec{v}_0)
&=& \sum_{T\subseteq V}(\prod_{\vec{v}\in T} \vec{v}) \nonumber
\end{eqnarray}
is the sum of all the $2^{|V|}$ distinct vectors spanned by $V$.
\end{proof}

\subsection{Proof of Theorem \ref{thm-1}}

\begin{proof}
We only need to show that algorithm FDTM is the desired algorithm.
As in \cite{bill13}, the correctness of algorithm FDTM is guaranteed by the nature of perfect hashing
and the correctness of the underlying group algebraic approach.
We focus on analyzing the time complexity of the algorithm.

Note that $q$ is a fixed constant.
By Naor {\em at el.}\cite{naor95}, Step 2 can be done in
$O(e^k k^{O(\log k)} \log^2 ((q-1)n)) = O^*(2.72^k)$
time, where $e = 2.718281\cdots $ is the natural constant. Step 3 can be easily done in $O(k^2)$ time.

It follows from Lemma \ref{rtm-lem3} that all those monomials that are not $q$-monomials in
$F$, and hence in $F'$, will be annihilated when variables $y_{ij}$
are replaced by $(\vec{v}_{\lambda(\tau(i,j))} + \vec{v}_0)$ in $G$ at Step 4.1.

Consider any given $q$-monomial $\pi$ of degree $k$ in $F$,
by Lemma \ref{rtm-lem2}, there are monomials
$\alpha\pi'$ in $F'$ such that the following are true:   $\alpha$ is a multilinear monomial of $z$-variables with degree $\le k(s(n)+1)+1$,
$\pi'$ is a degree $k$ multilinear monomial of $y$-variables, and
and all such monomials $\alpha$'s are distinct.
Since $\pi'$ has $k$ distinct $y$-variables, by the nature of perfect hashing functions,
there exists at least one hashing function $h \in \mathcal{H}$ that, at Step 4.1,
assigns each of the $k$ linearly independent vectors selected
at Step 3 to a unique $y$-variable in $\pi'$ to accomplish the replacements. At Step 4.2,
by Lemma \ref{rtm-lem4}, $\pi'$ (hence, $\alpha\pi$) will
survive the replacements at Step 4.1. Let $\mathcal{ S}$ be the
set of all the surviving $q$-monomials $\alpha\pi'$. Again,  by Lemma \ref{rtm-lem4}, we have
\begin{eqnarray}
G'&=& G(z_1,\ldots,z_h,(\vec{v}_{\lambda(\tau(1,1))}+\vec{v}_0),\ldots,(\vec{v}_{\lambda(\tau(1,q-1))}+\vec{v}_0),
 \nonumber \\
&& \hspace{7mm}
\ldots,(\vec{v}_{\lambda(\tau(n,1))}+\vec{v}_0),\ldots,(\vec{v}_{\lambda(\tau(n,q-1))}+\vec{v}_0)) \nonumber \\
&=&  \sum_{j=1}^{2^k} \left(\sum_{\alpha\pi'\in \mathcal{ S}} \alpha \right) \vec{v}_j 
 = \sum_{j=1}^{2^k} f_j(z_1,\ldots,z_h) \vec{v}_j 
 \not=  0 \nonumber
\end{eqnarray}
since $\mathcal{ S}$ is not empty.
Here, $f_j(z_1,\ldots,z_h)  =  \sum_{\alpha\pi'\in \mathcal{ S}} \beta.$
This means that, conditioned on that $\mathcal{ S}$ is not empty,
there is at least one $f_j$ that is not identical to zero.

For each given hashing function, we shall address the issues of how to calculate $G'$ and the
time needed to do so. Naturally, every element in the group
algebra $\mathcal{ F}[Z^k_2]$ can be represented by a vector in
$Z^{2^k}_2$. Adding two elements in $\mathcal{ F}[Z^k_2]$ is equivalent to
adding the two corresponding vectors in $Z_2^{2^k}$, and the
latter can be done in $O(2^k\log |\mathcal{F}|)$ time via component-wise sum.
In addition, multiplying two elements in $\mathcal{F}[Z^k_2]$ is
equivalent to multiplying the two corresponding vectors in
$Z_2^{2^k}$. By Lemma \ref{wh-lem2}, the latter can be done in $O(k2^{k+1}\log_2|\mathcal{F}|)$ with
the help of Walsh-Hadamard transformation.
By the circuit reconstruction and variable replacements
in Sections \ref{trans}, the size of the circuit $\mathcal{ C''}$ is at most
$O(n s(n)$. Calculating $G'$ by the circuit $\mathcal{ C''}$ consists of $O(n * s(n))$
arithmetic operations of either adding or multiplying two elements
in $\mathcal{ F}[Z^k_2]$ based on the circuit $\mathcal{ C''}$. Hence, the total
time needed is $O(n*s(n) k 2^{k+1}\log |\mathcal{F}|)$.

When $y$-variables are replaced by group algebraic elements at Step 4.1, according to the circuit reconstruction and
$x$-variable replacements in Section \ref{trans}, circuit $\mathcal{C}''$ is  $S$-read-once.
At Step 4, we run the deterministic algorithm of Lemma \ref{pit-lem} for $\mathcal{C}''$  to
simultaneously test whether there is one $f_j$ in $G'$ such that $f_j$ is
not identical to zero. The time needed for this step is $O^*(2^k s^4(n)\log |\mathcal{F}|)$, when a hashing function is given.

Recall that there are $e^k k^{O(\log k)} \log^2 ((q-1)n)) = O^*(2.72^k)$ many hashing functions in $\mathcal{H}$.
Recall also that $\log |\mathcal{F}| = log (k(s(n)+1)+1) +1$.
The total time for the entire
algorithm is  $O^*(2.72^k 2^k  s^4(n)) = O^*(5.44^k s^4(n))$.
\end{proof}

When the circuit size $s(n)$ is a polynomial in $n$, the time bound becomes $O^*(5.44^k).$

\subsection{Proof of Corollary \ref{cor}}

\begin{proof}
We follow the same approach as in proof for Theorem \ref{thm-1}. The difference is that
we treat $z$ as an algebraic symbol and additionally manipulate univariate polynomials of the single variable $z$ in
calculations at every node in the circuit $\mathcal{C}''$. Since we are interested in monomials of
a $z^t \pi$ format, where $\pi$ does not have $z$, the degrees of those polynomials can be bounded by $t$.
That is, if one such polynomial has a degree $>t$, then it can deleted along with its companion factor.
Adding two degree $t$ univariate polynomials needs time $O(t\log |\mathcal{F}|)$, and multiplying two of those polynomials
can be done in $O(t^2 \log |\mathcal{F}|)$. Hereby, it follows from the analysis for Theorem \ref{thm-1} that
the time bound for the corollary is $O^*(t^2 5.44^k s^4(n)).$
\end{proof}

\section{Concluding Comments}
The $5.44^k$ factor in the time bound for Theorem \ref{thm-1} and Corollary \ref{cor} is
contributed by the size of a family of perfect hashing functions  and the time needed for
the group algebraic approach to testing multilinear monomials. The size of the family of perfect hashing functions
by Naor {\em et al.} \cite{naor95} is near optimal. It is also known that the $O^*(2^k)$ time bound
for the algebraic approach to testing multilinear monomials is essentially optimal \cite{koutis09}.
Those two facts may imply that one could not reply on algebraic approach and perfect hashing functions
to improve our time bound. Nevertheless, since the lower bound in \cite{koutis09} is derived
for general circuits, it may be possible to improve the time bound for formulas.

\section*{Acknowledgment}

Shenshi is supported by Dr. Bin Fu's NSF CAREER Award, 2009 April 1 to 2014 March 31.

\end{document}